\documentclass[a4paper,11pt,openright]{article} 
\usepackage[english]{babel} 
\usepackage[latin1]{inputenc} 
\usepackage{amsmath,ascii,newunicodechar,amssymb,amsthm,amsfonts,amscd,authblk,dsfont,color,cancel,diagbox,enumerate,epsfig,eucal,fancyhdr,latexsym,lmodern,mathrsfs,mathtools,multicol,multirow,musicography,phaistos,skull,simpler-wick,subcaption,tikz-cd,secdot}
\usepackage[normalem]{ulem} 
\usepackage[multiuser]{fixme}
\FXRegisterAuthor{nd}{annd}{Nico} 
\FXRegisterAuthor{cvdv}{ancvdv}{Chris} 
\usetikzlibrary{shapes,snakes} 

\usepackage[
final=true,                    
plainpages=false,              
pdfpagelabels=true,            
pdfencoding=auto,              
unicode=true,                  
hypertexnames=true,            
naturalnames=true              
]{hyperref}
\usepackage{xr-hyper}
\usepackage[all,cmtip]{xy}
\usetikzlibrary{matrix,calc}

\definecolor{NiColor}{RGB}{77,77,255}
\definecolor{NiColoRed}{RGB}{255,77,77}
\definecolor{NiCitation}{RGB}{0,181,26}

\newtheoremstyle{TheoremStyle}
{3pt}
{3pt}
{\slshape}
{}
{\sc}
{:}
{.5em}
{}

\theoremstyle{TheoremStyle}
\newtheorem{theorem}{Theorem}

\newtheorem{lemma}[theorem]{Lemma}
\newtheorem{definition}[theorem]{Definition}
\newtheorem{remark}[theorem]{Remark}

\newtheorem{example}[theorem]{Example} 

\makeatletter
\def\@endtheorem{\hfill$\lozenge$} 
\makeatother

\title{Global observables in statistical mechanics}

\author{\href{mailto:chris.vandeven@ru.nl}{C. J. F. van de Ven}}
\affil[]{Radboud University, Faculty of Science, Heyendaalse weg 13, 56525 AJ Nijmegen,The Netherlands}

\def\bearray{\begin{eqnarray}}
\def\earray{\end{eqnarray}}
\def\beq{\begin{equation}}
\def\eeq{\end{equation}}

\def\b0{{\bf 0}}

\newsymbol\bt 1202             

\newsymbol\rest 1316         

\def\b{\textnormal{b}}

\begin{document}

\maketitle
\begin{abstract}
\noindent
We present a canonical construction of global observables -sometimes referred to in the literature as macroscopic observables or observables at infinity- in statistical mechanics, providing a unified treatment of both commutative and non-commutative cases. Unlike standard approaches, the framework is formulated directly in the $C^*$-algebraic setting, without relying on any specific representation. 
\end{abstract}

\section{Introduction}
In lattice statistical mechanics one studies large collections of interacting particles arranged on a (regular) lattice, and analyzes their behavior in appropriate limiting regimes to uncover their macroscopic properties.
Two distinct types of such ``infinite volume'' limits are commonly considered. 

The first is the \emph{thermodynamic limit},
in which the lattice size tends to infinity while local observables are examined within finite regions of the lattice. This limit is mathematically formalized by the so-called quasi-local algebra and allows one to rigorously define expectation values of local quantities and ensures that bulk properties, such as energy density or correlation functions, stabilize in the infinite-volume system. Its significance is well established in the context of phase transitions and equilibrium phenomena in quantum and classical statistical mechanics \cite{Bratteli_Robinson_1981_I,Bratteli_Robinson_1981_II, Israel}.  

The second, less standard but equally important, is what is sometime referred to as \emph{macroscopic limit} \cite{DRW,Landsman_2017}, which focuses on observables that are  often called ``global'' or ``observables at infinity,'' such as spatial averages of local quantities over regions of diverging size \cite{Enter,Hepp,LR}. Unlike the thermodynamic limit, which addresses the stability of local observables under volume growth, the macroscopic limit captures collective, non-local behavior that reflects emergent classical features of quantum systems \cite{DV2,LMV}. While this construction is classical in spirit and well established in classical statistical mechanics \cite{Georgii}, its quantum counterpart is considerably less developed, often requiring sophisticated representation-theoretic methods \cite{Bratteli_Robinson_1981_I,Bratteli_Robinson_1981_II}.

In this work we develop a $C^*$-algebraic framework to study the macroscopic behavior of statistical mechanical systems, providing a unified treatment of both commutative and non-commutative settings. This approach emphasizes the  connection between quantum and classical systems and does not require the use of any particular representation. 
In the non-commutative setting, the resulting $C^*$-algebra naturally contains  commutative $C^*-$ subalgebras generated by macroscopic averages, thereby providing a bridge between classical thermodynamics as a limit of quantum statistical mechanics, confirming in particular the ideas described in \cite{Landsman_2017,Raggio}.

. 

\section{Full $C^*$-product}
Let $\Gamma \subset \mathbb{R}^d$ be a countable set. This already endows the set of finite subsets of $\Gamma$ with a partial order (inclusion), which is upward directed and hence defines a notion of convergence, namely ``convergence along the net of finite subsets of $\Gamma$, directed by inclusion', which we denote by
\[
\lim_{\Lambda \nearrow \Gamma} F(\Lambda) = \alpha.
\]
This means that for every $\varepsilon > 0$, there exists a finite subset $K_\varepsilon \Subset \Gamma$ such that
\[
\|F(\Lambda) - \alpha\| < \varepsilon \quad \text{for all } \Lambda \supset K_\varepsilon.
\]
Here, $F$ is a function (or operator) defined on finite subsets of $\Gamma$ and taking values in a normed space.\footnote{This notion of convergence should not be confused with convergence in the sense of van Hove.}
\\\\
To each $x \in \Gamma$ we associate a unital $C^*$-algebra $\mathcal{B}_x$; in fact, we assume the same algebra $\mathcal{B}$ for all $x$, and use $x$ only to denote the lattice position. The (minimal) tensor product of $\otimes_{x\in\Lambda}\mathcal{B}_x$ is denoted by $\mathcal{B}_\Lambda$ and ensures that $\mathcal{B}_\Lambda$ is again a $C^*$-algebra. 
We  consider the \textit{full $C^*$-product} $\prod_{\Lambda\Subset\Gamma}\mathcal{B}_\Lambda$ defined by
\begin{align}\label{Eq: full C*-product}
\prod_{\Lambda\Subset\Gamma}\mathcal{B}_\Lambda:=\{(a_\Lambda)_\Lambda\,|\,(\|a_\Lambda \|_\Lambda)_\Lambda\in\ell^\infty(\{\Lambda\Subset\Gamma\})\}\,,
\end{align}
where $(a_\Lambda)_\Lambda$ should be understood as element in the algebraic direct product with pointwise operations, and $\{\Lambda\Subset\Gamma\}$ means the set of all finite subsets of $\Gamma$. 
As it is well-known \cite{Blackadar_2006} $\prod_{\Lambda\Subset\Gamma}\mathcal{B}_\Lambda$ is a $C^*$-algebra with respect to sup norm $\|(a_\Lambda)_\Lambda\|_\infty:=\sup\limits_{\Lambda\Subset\Gamma}\|a_\Lambda\|_\Lambda$.
For $(a_\Lambda)_\Lambda,(b_\Lambda)_\Lambda\in \Pi_{\Lambda\Subset\Gamma}\mathcal{B}_\Lambda$, we introduce the following $\sim$-equivalence relation 
 \begin{align}\label{Eq: equivalence relation among sequences}
(a_\Lambda)_\Lambda\sim(b_\Lambda)_\Lambda\Longleftrightarrow
 	\lim_{\Lambda\nearrow\Gamma}\|a_\Lambda-b_\Lambda\|_\Lambda=0\,.
 \end{align}
 For each given sequence $(a_\Lambda)_\Lambda$, we will denote by $[a_\Lambda]_\Lambda:=[(a_\Lambda)_\Lambda]$ the corresponding equivalence class with respect to \eqref{Eq: equivalence relation among sequences}.
 Moreover, the \textit{direct $C^*$-sum}
 \begin{align}\label{Eq: direct C*-sum}
 \bigoplus_{\Lambda\Subset\Gamma}\mathcal{B}_\Lambda:&=\{(a_\Lambda)_\Lambda\in\prod_{\Lambda\Subset\Gamma}\mathcal{B}_\Lambda\,|\,
 (\|a_\Lambda\|_\Lambda )_\Lambda\in C_0(\{\Lambda\Subset\Gamma\})\}\\&=\{(a_\Lambda)_\Lambda\in\prod_{\Lambda\Subset\Gamma}\mathcal{B}_\Lambda\,|\,
 \lim_{\Lambda\nearrow\Gamma}\|a_\Lambda\|_\Lambda=0\}\,
 \end{align}
 is a closed two-sided ideal in $\prod_{\Lambda\Subset\Gamma}\mathcal{B}_\Lambda$ and thus we may consider the quotient
 \begin{align}\label{Eq: quotient algebra}
[\mathcal{B}]_\sim:=\prod_{\Lambda\Subset\Gamma}\mathcal{B}_\Lambda/\bigoplus_{\Lambda\Subset\Gamma}\mathcal{B}_\Lambda\,,
 \end{align}
 which is nothing but the space of $\sim$-equivalence classes $[(a_\Lambda)_\Lambda]_\Lambda$ for bounded sequences $(a_\Lambda)_\Lambda$, i.e.
$$[\mathcal{B}]_\sim=q(\prod_{\Lambda\Subset\Gamma}\mathcal{B}_\Lambda)$$
 where $q$ is the canonical quotient map, 
 \[
 q : \prod_{\Lambda\Subset\Gamma} \mathcal{B}_\Lambda \;\longrightarrow\; \prod_{\Lambda\Subset\Gamma} \mathcal{B}_\Lambda/\bigoplus_{\Lambda\Subset\Gamma} \mathcal{B}_\Lambda, \quad
 q\big( (a_\Lambda)_{\Lambda\Subset\Gamma} \big) = [ (a_\Lambda)_{\Lambda \Subset \Gamma} ]_\Lambda.
 \]
 Importantly, $[\mathcal{B}]_\sim$ is a $C^*$-algebra with norm
 \begin{align}\label{Eq: norm of the quotient algebra}
 	\|[a_\Lambda]_\Lambda\|_{[\mathcal{B}]_\sim}=\limsup_{\Lambda\nearrow \Gamma}\|a_\Lambda\|_\Lambda\,.
 \end{align}
 Hence, passing to the quotient makes it possible to identify sequences that represent the same observable in the limit as $\Lambda\nearrow\Gamma$, so that only essentially distinct  observables are captured. Of course, $[\mathcal{B}]_\sim$  is very large, so that one typically focuses on suitable $C^*$-subalgebras.
\begin{example}[Quasi-local algebra]
{\em
We consider the $*$-algebra of local observables
\[
\dot{\mathcal{B}}^\infty := \bigcup_{\Lambda \Subset \Gamma} \mathcal{B}_\Lambda
\]
constructed as follows. Each local algebra $\mathcal{B}_{\Lambda_0}$ is embedded into larger volumes via the canonical unital $*$-homomorphisms
\[
\iota_{\Lambda_0,\Lambda}(b_{\Lambda_0})
:= b_{\Lambda_0} \otimes \mathbf{1}_{\Lambda \setminus \Lambda_0},
\qquad \Lambda_0 \subset \Lambda,
\]
where $\mathbf{1}_{\Lambda \setminus \Lambda_0}$ denotes the identity element in
$
\mathcal{B}_{\Lambda \setminus \Lambda_0}
= \bigotimes_{x \in \Lambda \setminus \Lambda_0} \mathcal{B}_x.
$
These maps turn $\{\mathcal{B}_\Lambda\}_{\Lambda \Subset \Gamma}$ into an inductive system of $*$-algebras. In particular, $\dot{\mathcal{B}}^\infty$ is a $*$-algebra under the induced operations and contains all local observables.
Via the induced embeddings, $\dot{\mathcal{B}}^\infty$ can be realized as a $*$-subalgebra of the full $C^*$-product
$\prod_{\Lambda \Subset \Gamma} \mathcal{B}_\Lambda$.
\\\\
The {\bf quasi-local algebra} is then defined to be the completion of the quotient
$$[\mathcal{B}]^\infty:=\overline{q(\dot{\mathcal{B}}^\infty)}^{\|\cdot\|}\subset [\mathcal{B}]_\sim,$$
where  norm is given by $\|[a_\Lambda]_\Lambda\|:=\|[a_\Lambda]_\Lambda\|_{[\mathcal{B}]^\infty}=\limsup_{\Lambda\nearrow\Gamma}\|a_\Lambda\|_\Lambda$, which can be shown to equal the actual limit. 
It is not difficult to see that $[\mathcal{B}]^\infty$ is a $C^*$-subalgebra of $[\mathcal{B}]_\sim$. Note that this construction precisely coincides with the usual definition of the quasi-local algebra, as explained in e.g. \cite[Chapter 8]{Landsman_2017}. The quasi-local algebra provides the standard framework for describing the thermodynamic limit in statistical mechanics \cite{Bratteli_Robinson_1981_I,Bratteli_Robinson_1981_II}.
\hfill$\diamond$
}
\end{example}

\noindent
The  $C^*$-algebra $\prod_{\Lambda\Subset\Gamma} \mathcal{B}_\Lambda$ contains many other elements for which the expectation value is not defined in several physically relevant states, such as translation-invariant states. 
We give some examples of such sequences.
\begin{example}
We consider the following sequences of tensor products of matrices:
\begin{itemize}
    \item[(i)]
\[
a_\Lambda := \bigotimes_{x \in \Lambda} \sigma^1_x,
\]
where $\sigma^1_x$ is the Pauli matrix $\sigma^1$ at site $x$.  
\item[(ii)] 
\[
a_\Lambda := \bigotimes_{x \in \Lambda} 
\begin{cases}
\sigma^1_x, & x \text{ odd},\\
\sigma^3_x, & x \text{ even}.
\end{cases}
\]
\item[(iii)] Define a sequence $(a_\Lambda)_\Lambda$ by partitioning $\Gamma$ into consecutive blocks of strictly increasing lengths $(B_n)_{n\ge0}$, assigning $\sigma^1$ to even-indexed blocks and $\sigma^3$ to odd-indexed blocks, and setting
\[
a_\Lambda := \bigotimes_{x \in \Lambda} a_x, \quad 
a_x :=
\begin{cases}
\sigma^1 & \text{if $x$ is in an even-indexed block,}\\
\sigma^3 & \text{if $x$ is in an odd-indexed block.}
\end{cases}
\]
\end{itemize}
For a translation-invariant state $\omega$, the limit 
\[
\lim_{\Lambda \nearrow \Gamma} \omega(a_\Lambda)
\]
is generally not defined for these sequences, because the product over infinitely many fluctuating operators diverges or oscillates.
\hfill$\diamond$
\end{example}

\section{$C^*$-algebraic construction of global observables in quantum statistical mechanics}
We now construct another $C^*$-subalgebra of $\prod_{\Lambda\Subset\Gamma} \mathcal{B}_\Lambda$.
Let
$$
\mathcal{C}^\infty:=\left\{ (a_\Lambda)_\Lambda \in \prod_{\Lambda\Subset\Gamma} \mathcal{B}_\Lambda \ :\  \forall \ (b_\Lambda)_\Lambda \in \dot{ \mathcal{B}}^\infty \ \lim_{\Lambda \nearrow \Gamma} \|[a_\Lambda, b_\Lambda]\|_\Lambda = 0 \right\}.
$$
\begin{lemma}
$\mathcal{C}^\infty$ is a norm-closed *-subalgebra of $\prod_\Lambda \mathcal B_\Lambda$, and hence a $C^*$-algebra.
\end{lemma}
\begin{proof}
We must prove the following three conditions.
\begin{itemize}
    \item[(1)] $*$-algebra property: Let $a=(a_\Lambda)_\Lambda, a'=(a_\Lambda')_{\Lambda}\in \mathcal C^\infty$ and $b=(b_\Lambda)_\Lambda\in \dot{\mathcal B}^\infty$. Then
\[
[a+a',b] = [a,b] + [a',b], \quad [a^*,b] = -[a,b^*]^*, \quad [aa',b] = a[a',b] + [a,b]a'.
\]
Since $\|[a_\Lambda,b_\Lambda]\|\to 0$ and $\|[a'_\Lambda,b_\Lambda]\|\to 0$, the same holds for $a+a'$, $a^*$, and $aa'$. Thus $\mathcal C^\infty$ is a *-subalgebra.
\item[(2)] Norm-closedness: Let $(a^{(\iota)})\subset \mathcal C^\infty$ be a net that converges in the sup-norm to $a$, i.e.
\(\|a^{(\iota)}-a\|_\infty \to 0\). Then for any $b\in \dot{\mathcal B}^\infty$ and each $\Lambda$,
\[
\|[a_\Lambda,b_\Lambda]\|_\Lambda \le \|[a_\Lambda - a^{(\iota)}_\Lambda, b_\Lambda]\|_\Lambda + \|[a^{(\iota)}_\Lambda,b_\Lambda]\|_\Lambda
\le 2 \|a - a^{(\iota)}\|_\infty \|b_\Lambda\|_\Lambda + \|[a^{(\iota)}_\Lambda,b_\Lambda]\|_\Lambda.
\]
Taking the limit $\Lambda\nearrow \Gamma$ and then $\iota\to\infty$, the right-hand side goes to $0$. Hence $a\in \mathcal C^\infty$.
\item[(3)] $C^*$-property: Being a closed *-subalgebra of the C*-algebra $\prod_\Lambda \mathcal B_\Lambda$ with the sup-norm, $\mathcal C^\infty$ is itself a $C^*$-algebra.
\end{itemize}
\end{proof}
\noindent
Since $\bigoplus_{\Lambda\Subset\Gamma}\mathcal{B}_\Lambda\subset \mathcal{C}^\infty$, we can define
\[
[\mathcal{C}]^\infty:=\mathcal{C}^\infty/
\bigoplus_{\Lambda\Subset\Gamma}\mathcal{B}_\Lambda
=q(\mathcal{C}^\infty).
\]
It follows that $[\mathcal{C}]^\infty$ is a $C^*$-algebra, and 
$$[\mathcal{C}]^\infty=[\mathcal{B}]_\sim\cap (q(\dot{\mathcal{B}}^\infty))'$$
is the relative commutant of the $*$-algebra $[\dot{\mathcal{B}}]^\infty$ in the  $C^*$-algebra $[\mathcal{B}]_\sim$, and therefore norm-closed. 
Furthermore, we may define
$$
\mathcal{C}_\Lambda^\infty:=\left\{ (a_{\Lambda'})_{\Lambda'} \in \prod_{\Lambda'\Subset\Gamma} \mathcal{B}_{\Lambda'} : \forall \quad b_\Lambda\in  \mathcal{B}_\Lambda, \ \lim_{\Lambda' \nearrow \Gamma} \|[a_{\Lambda'}, b_\Lambda]\|_{\Lambda'} = 0 \right\},
$$
where $b_\Lambda$ has to be understood as element of $\mathcal{B}_{\Lambda'}$ under the canonical embedding. It can be demonstrated that $\mathcal{C}_\Lambda^\infty$ is a closed $C^*$-subalgebra of the full $C^*$-product and $q(\mathcal{C}_\Lambda^\infty)$ is closed in $[\mathcal{B}]_\sim$. This implies that 
$$
\mathcal{C}^\infty=\bigcap_{\Lambda \Subset \Gamma} \mathcal{C}_\Lambda^\infty.
$$

\begin{remark}
We would like to point out that the $C^*$-algebra $\mathcal{C}^\infty$
depends only on the net of local $C^*$-algebras over the directed set of finite subsets of $\Gamma$, and not on any particular exhaustion of $\Gamma$. 
Indeed, the  elements
$(a_\Lambda)_\Lambda$ of the full $C^*$-product $\prod_{\Lambda \Subset \Gamma} \mathcal{B}_\Lambda$
are indexed by \emph{all} finite subsets of $\Gamma$ simultaneously, and the
defining condition
\[
  (a_\Lambda)_\Lambda \in \mathcal{C}^\infty
  \quad \Longleftrightarrow \quad
  \forall\, (b_\Lambda)_\Lambda \in \dot{\mathcal{B}}^\infty,\quad
  \lim_{\Lambda \nearrow \Gamma} \|[a_\Lambda, b_\Lambda]\|_\Lambda = 0.
\]
\hfill$\diamond$
\end{remark}
\noindent
Motivated by the pioneering ideas described in \cite[Sec.~2.3.6]{Enter}, we introduce the following definition.
\begin{definition}
The $C^*$-subalgebra $[\mathcal{C}]^\infty$ is called the \textbf{algebra of global observables}. 
\hfill$\diamond$
\end{definition}
\noindent
It can be proved  that $[\mathcal{C}]^\infty$ is not separable.\footnote{The rigorous proof of this fact is quite long and technical and is left for future work. It follows from the existence of uncountably many pairwise distance-separated elements in $\mathcal{C}^\infty$.} Moreover, it turns out to be non-commutative, as the following lemma shows.
\begin{lemma}\label{Lemm: non-comm}
  $[\mathcal{C}]^\infty$ is non-commutative.
\end{lemma}
\begin{proof}
Consider the sequences $a=(a_\Lambda)_\Lambda$ and $c=(c_\Lambda)_\Lambda$, defined by
\[
a_\Lambda := \mathbf{1}_{\Lambda \setminus \{x_\Lambda\}} \otimes a_{x_\Lambda}, \quad
c_\Lambda := \mathbf{1}_{\Lambda \setminus \{x_\Lambda\}} \otimes c_{x_\Lambda},
\]
where \(a_{x_\Lambda}, c_{x_\Lambda}\) are fixed non-commuting observables acting on site \(x_\Lambda \in \Lambda\), i.e., local observables in \(\mathcal{B}_{\{x_\Lambda\}}\), such that \([a_{x_\Lambda}, c_{x_\Lambda}] \neq 0\). These sequences are what we call ``local sequences translated to infinity,'' where \(x_\Lambda\) is chosen so that \(x_\Lambda \to \infty\) as \(\Lambda \nearrow \Gamma\).
\\\\
For any local observable \(b \in \dot{\mathcal{B}}^\infty\), supported on a finite subset \(\Lambda_0 \Subset \Gamma\), define the sequence
\[
b_\Lambda := 
\begin{cases}
b_{\Lambda_0} \otimes \mathbf{1}_{\Lambda \setminus \Lambda_0}, & \Lambda_0 \subseteq \Lambda, \\
0, & \text{otherwise}.
\end{cases}
\]
Since \(a_\Lambda\) acts non-trivially only on site \(x_\Lambda\) which eventually is supported outside  \(\Lambda_0\) and \(b_\Lambda\) acts non-trivially only on sites in \(\Lambda_0\), their supports are disjoint for all sufficiently large \(\Lambda\). Hence,
\[
[a_\Lambda, b_\Lambda] = 0 \quad \text{for all sufficiently large } \Lambda.
\]
Therefore,
\[
\lim_{\Lambda \nearrow \Gamma} \|[a_\Lambda, b_\Lambda]\|_\Lambda = 0,
\]
and similarly for \(c_\Lambda\). This shows
$(a_\Lambda)_\Lambda, (c_\Lambda)_\Lambda \in \mathcal{C}^\infty$.
By definition,
\[
[a_\Lambda, c_\Lambda] 
= (\mathbf{1}_{\Lambda \setminus \{x_\Lambda\}} \otimes a_{x_\Lambda})(\mathbf{1}_{\Lambda \setminus \{x_\Lambda\}} \otimes c_{x_\Lambda}) 
- (\mathbf{1}_{\Lambda \setminus \{x_\Lambda\}} \otimes c_{x_\Lambda})(\mathbf{1}_{\Lambda \setminus \{x_\Lambda\}} \otimes a_{x_\Lambda})=\mathbf{1}_{\Lambda \setminus \{x_\Lambda\}} \otimes [a_{x_\Lambda}, c_{x_\Lambda}],
\]
which implies
\[
\|[a_\Lambda, c_\Lambda]\|_\Lambda = \|[a_{x_\Lambda}, c_{x_\Lambda}]\|_{x_\Lambda}.
\]
Note that the norm on the right-hand side does not depend on \(\Lambda\). A similar computation shows that the same holds true for equivalence classes $[a]$ and $[c]$. Recall that the norm on the quotient algebra \([\mathcal{C}]^\infty\) is defined as
 \[
 \|[[a], [c]]\| := \limsup_{\Lambda \nearrow \Gamma} \|[a_\Lambda, c_\Lambda]\|_\Lambda.
\]
 Since \(\|[a_\Lambda, c_\Lambda]\|_\Lambda = \|[a_{x_\Lambda}, c_{x_\Lambda}]\|_{x_\Lambda}\) for all \(\Lambda\), and the since the commutator \([a_{x_\Lambda}, c_{x_\Lambda}] \neq 0\), we conclude that
\[
\|[[a], [c]]\| > 0.
\]
As a result, the algebra \([\mathcal{C}]^\infty\) is non-commutative.
\end{proof}
\noindent
The observables ``translated to infinity'' considered in Lemma~\ref{Lemm: non-comm} play a role in e.g. measuring correlations, as these are precisely the expectation values of products of local observables and their large-distance translates, and their asymptotic behavior encodes the thermodynamic phase structure of the system.   Besides these observables we give another example of a global observable.
\begin{example}
Set $\Gamma=\mathbb{Z}$.
Consider the following sequence of local tensor products:
\[
c_N := \bigotimes_{k=1}^{N - \lfloor N/2 \rfloor} \mathbf{1}\;\otimes\; \bigotimes_{j=1}^{\lfloor N/2 \rfloor} a,
\]
where $\mathbf{1}$ denotes the identity operator on a single site and $a$ is a non-trivial and non-zero single-site matrix with $\|a\|\le 1$. This sequence looks like
\[
\begin{aligned}
c_1 &= \mathbf{1}, \\
c_2 &= \mathbf{1} \otimes a, \\
c_3 &= \mathbf{1} \otimes \mathbf{1} \otimes a, \\
c_4 &= \mathbf{1} \otimes \mathbf{1} \otimes a \otimes a, \\
c_5 &= \mathbf{1} \otimes \mathbf{1} \otimes \mathbf{1} \otimes a \otimes a, \\
c_6 &= \mathbf{1} \otimes \mathbf{1} \otimes \mathbf{1} \otimes a \otimes a \otimes a, \\
c_7 &= \mathbf{1} \otimes \mathbf{1} \otimes \mathbf{1} \otimes \mathbf{1} \otimes a \otimes a \otimes a, \\
c_8 &= \mathbf{1} \otimes \mathbf{1} \otimes \mathbf{1} \otimes \mathbf{1} \otimes a \otimes a \otimes a \otimes a.
\end{aligned}
\]
Note that the sequence $c=(c_N)_N$ is uniformly bounded, and for any fixed local observable $b$, $c$ asymptotically commutes with $b$.
Therefore, $c$ defines a legitimate element in $\mathcal{C}^\infty$. However $c$ does in general not commute with the sequences constructed in the  proof Lemma \ref{Lemm: non-comm}. An easy computation shows that the same remains true for $[c]\in [\mathcal{C}]^\infty$.
\hfill$\diamond$
\end{example}
\noindent
This raises the natural
question of characterizing the center $\mathcal{Z}([\mathcal{C}]^\infty)$,
which is a proper commutative $C^*$-subalgebra of $[\mathcal{C}]^\infty$. This is a non-trivial issue that will be addressed in a follow-up work.

\subsection{Comparison with current approaches}
In this section we briefly recall the well-known state-dependent construction of the so called ``center'' and the ``algebra at infinity'' \cite{Bratteli_Robinson_1981_I}.
To each finite region $\Lambda\Subset\Gamma$, the complement $\Lambda^c:=\Gamma\setminus\Lambda$ induces the subalgebra of the quasi-local algebra $[\mathcal{B}]^\infty$, given by
\begin{align*}
[\mathcal{B}_{\Lambda^c}]=\overline{q\bigg{(}\bigcup_{\Delta\Subset\Lambda^c}\mathcal{B}_{\Delta}\bigg{)}}^{\|\cdot\|},
\end{align*}
where $\mathcal{B}_{\Lambda^c}:=\bigcup_{\Delta\Subset\Lambda^c}\mathcal{B}_{\Delta}\subset$ $\dot{\mathcal{B}}^\infty$. Subsequently, one may assign to a given state $\omega$ on $[\mathcal{B}]^\infty$  two $C^*$-algebras:
\begin{itemize}
    \item the center $\mathcal{B}_\omega^c:=(\pi_\omega([\mathcal{B}]^\infty))''\cap (\pi_\omega([\mathcal{B}]^\infty))'$;
    \item the algebra at infinity $\mathcal{B}_\omega^\infty:=\bigcap_{\Lambda\Subset\Gamma}(\pi_\omega([\mathcal{B}_{\Lambda^c}]))''$,
\end{itemize}
where $\pi_\omega$ is the GNS-representation with GNS Hilbert space $\mathcal{H}_\omega$ and $(\pi_\omega([\mathcal{B}]^\infty))'$ and $(\pi_\omega([\mathcal{B}]^\infty))''$ denote the commutant, resp. double commutant in $B(\mathcal{H}_\omega)$, i.e.  the algebra of bounded operators on $\mathcal{H}_\omega$. In particular, the algebra at infinity contains macroscopic averages (see Section \ref{sec:subalgebras}) and if the algebra at infinity is trivial, the state $\omega$ can be identified with a pure thermodynamic phase, see e.g. \cite[Chapter 8]{Landsman_2017} for a comprehensive discussion.
\\\\
To make a comparison with the current setting as developed in this note, we notice
that the algebra $[\mathcal{B}_{\Lambda^c}]$ can be rewritten as
$$
[\mathcal{B}_{\Lambda^c}]:=\overline{\left\{ (a_{\Lambda'})_{\Lambda'} \in \dot{\mathcal{B}}^\infty : \forall \quad b_\Lambda\in  \mathcal{B}_\Lambda, \ \lim_{\Lambda' \nearrow \Gamma} \|[a_{\Lambda'}, b_\Lambda]\|_{\Lambda'} = 0 \right\}/\bigoplus_{\Lambda\Subset\Gamma}\mathcal{B}_\Lambda}^{\|\cdot\|},
$$
which is the analog of $[C_\Lambda]^\infty$ in $[\mathcal{B}]^\infty$. In particular, $[\mathcal{B}_{\Lambda^c}]\subset [\mathcal{C}_\Lambda]^\infty$, where the inclusion can be shown to be strict if the single site algebra is separable.  Indeed, the latter implies that
$[\mathcal{B}_{\Lambda^c}]$ is separable but  $[\mathcal{C}_\Lambda]^\infty$ is {\em not}. As a result, $[\mathcal{C}]^\infty$ has many  more observables affiliated with it, and since it is state independent,  is a natural candidate for describing global observables in  statistical mechanics.

\subsection{Commutative subalgebras}\label{sec:subalgebras}
We now construct a $C^*$-subalgebra of
$[\mathcal{C}]^\infty$, containing all ``macroscopic averages". This algebra induces a commutative $C^*-$algebra by passing to a suitable quotient, and therefore resembles 
a classical observable algebra describing the macroscopic limit arising from the underlying quantum statistical mechanics.
\\\\
To illustrate this idea, we focus on the one-dimensional lattice i.e., $\Gamma = \mathbb{Z}$.  Let $\Lambda_1\subset \Lambda_2\subset\dots$ be a strictly increasing sequence of connected finite subsets of $\mathbb{Z}$, with $|\Lambda_N|=N$ and $\cup_N\Lambda_N=\mathbb{Z}$.
For each region $\Lambda=\Lambda_N$, we consider the linear operator (left-shift operator) 
\[
\gamma_\Lambda \colon \mathcal{B}_\Lambda \to \mathcal{B}_\Lambda,
\]
uniquely defined by continuous and linear extension of the following map on elementary tensors:
\begin{align}\label{Eq: gamma-operator}
	\gamma_{\Lambda}(a_{1} \otimes \ldots \otimes a_{N}) := a_{2} \otimes \ldots \otimes a_{N} \otimes a_{1}\,,
\end{align}
where \(a_{1}, \dots, a_{N} \in \mathcal{B}\).
The operator $\gamma_\Lambda$ is a $*$-endomorphism of the algebra $\mathcal{B}_\Lambda$. Moreover, $\gamma_\Lambda^{N} = \operatorname{id}$. We then  define the \textbf{averaged shift map}
\begin{align}\label{Eq: averaged gamma-operator}
	\overline{\gamma}_\Lambda := \frac{1}{N} \sum_{j=0}^{N-1} \gamma_\Lambda^j\,,
\end{align}
where $\gamma_\Lambda^j=\gamma_\Lambda\circ \cdots \circ \gamma_\Lambda$ ($j$-times).
The image of this map, i.e.
$\overline{\gamma}_\Lambda(\mathcal{B}_\Lambda)$
is a $C^*$-subalgebra of $\mathcal{B}_\Lambda$ consisting of the $\gamma$-invariant elements.

\begin{definition}\label{Def: gamma-sequence}
A sequence \((a_\Lambda)_\Lambda\) is called a \textbf{$\gamma$-sequence} if there exists a finite subset \(\Lambda_0 \Subset \Gamma\) and an element \(a_{\Lambda_0} \in \mathcal{B}_{\Lambda_0}\) such that
\begin{align}\label{Eq: gamma-sequence property}
	a_\Lambda = \overline{\gamma}_\Lambda^{\Lambda_0}(a_{\Lambda_0}) 
	:= 
	\begin{cases}
		\overline{\gamma}_\Lambda\left( \mathbf{1}_{\Lambda \setminus \Lambda_0} \otimes a_{\Lambda_0} \right), & \Lambda \supseteq \Lambda_0, \\
		0, & \text{otherwise},
	\end{cases}
\end{align}
where \(\mathbf{1}_{\Lambda \setminus \Lambda_0}\) denotes the identity on the complementary tensor factors.
\hfill$\diamond$
\end{definition}
In what follows, we define the $*$-algebra \(\dot{\mathcal{B}}^\infty_\gamma \subset \prod_{\Lambda \Subset \Gamma} \mathcal{B}_\Lambda\) generated by all $\gamma$-sequences, together with its image under the canonical projection \([\dot{\mathcal{B}}]^\infty_\gamma \subset [\mathcal{B}]_\sim\).  The algebra $[\dot{\mathcal{B}}]^\infty_\gamma$ enjoys remarkable properties; in particular, it can be completed to a commutative \(C^*\)-algebra \([\mathcal{B}]^\infty_\gamma\) \cite[Prop.~6]{DV2}. Moreover, the $\gamma$-sequences constitute a continuous field of $C^*$-algebras over $\mathbb{N}\cup\{\infty\}$ with limit $[\mathcal{B}]^\infty_\gamma$.

\begin{remark}\label{Remark: wichtig}
 {\em 
Each state on  $[\mathcal{B}]^\infty_\gamma$ canonically induces a translationally invariant state on the quasi-local algebra $[\mathcal{B}]^\infty$. Indeed, it is not difficult to see that the map
\begin{align*}
    &\overline{\gamma}_\infty: [\mathcal{B}]^\infty\to [\mathcal{B}]^\infty_\gamma;\\
    &[a_\Lambda]_\Lambda\mapsto [\overline{\gamma}_\Lambda(a_{\Lambda})]_\Lambda
\end{align*}
is well-defined. In particular, applying a state $\omega\in S([\mathcal{B}]^\infty_\gamma)$ to an equivalence class of $\gamma$-sequences, yields
$$
\omega([\overline{\gamma}_\Lambda(a_{\Lambda})])=\omega\circ \overline{\gamma}_\infty([a_\Lambda]_\Lambda)
$$
If now we define 
$$
\hat{\omega}:=\omega\circ \overline{\gamma}_\infty,
$$
then $\hat{\omega}$ is a translationally invariant state on  $[\mathcal{B}]^\infty$, since for any $j$, one has
$\hat{\omega}\circ \gamma_\infty^j=\hat{\omega}$,
where 
$$
\gamma_\infty^j([a_\Lambda]_\Lambda)=[\gamma_\Lambda^j(a_\Lambda)]_\Lambda
$$
is the translation of $j$ lattice sites.
\\\\
\noindent
In this way, the expectation values of $\gamma$-sequences in $\hat{\omega}$ capture the macroscopic averages of local observables and thereby encode the distribution over ergodic components of the corresponding translation-invariant state, see  \cite[Chapter 4]{Bratteli_Robinson_1981_I} for further details hereon.
\hfill$\diamond$ 
}
\end{remark}
\noindent
Moreover, as the following lemma shows $[\mathcal{B}]^\infty_\gamma$ is a $C^*$-subalgebra of $[\mathcal{C}]^\infty$.
 \begin{lemma}
     $[\mathcal{B}]^\infty_\gamma\subset [\mathcal{C}]^\infty$.
 \end{lemma}
 \begin{proof}
 The assertion follows directly from the construction. We first prove it for $\gamma$-sequences; 
the general case then follows immediately from \cite[Prop.~6]{DV2}. 
For any $\gamma$-sequence $(a_\Lambda)_\Lambda$ as in Definition~\ref{Def: gamma-sequence}, 
we claim that for all local observables $(b_\Lambda)_\Lambda \in \dot{\mathcal{B}}^\infty$ 
\begin{equation}
  \lim_{\Lambda \nearrow \Gamma} \|[a_\Lambda, b_\Lambda]\|_\Lambda = 0.
\end{equation}
To see this, fix $\Lambda_0 \Subset \Gamma$ and $a_{\Lambda_0} \in \mathcal{B}_{\Lambda_0}$ such that
\begin{equation}
  a_\Lambda = \overline{\gamma}_{\Lambda}^{\Lambda_0}(a_{\Lambda_0})
  = \frac{1}{N} \sum_{j=0}^{N-1}
    \gamma_\Lambda^j \big( 1_{\Lambda \setminus \Lambda_0} \otimes a_{\Lambda_0} \big),
  \qquad \Lambda \supseteq \Lambda_0,
\end{equation}
and $a_\Lambda = 0$ otherwise. Let $(b_\Lambda)_\Lambda$ be a local observable supported on some fixed $\Lambda' \Subset \Gamma$. 
Since $b_\Lambda$ acts non-trivially only on the tensor factors associated with $\Lambda'$, 
and $a_\Lambda$ is an average over cyclic shifts, the commutator norm can be estimated as
\begin{align}
  \| [a_\Lambda, b_\Lambda] \|_\Lambda
  &= \left\| \frac{1}{N} \sum_{j=0}^{N-1}
    \left[ \gamma_\Lambda^j \big( 1_{\Lambda \setminus \Lambda_0} \otimes a_{\Lambda_0} \big),
           b_\Lambda \right] \right\|_\Lambda \\
  &\leq \frac{1}{N} \sum_{j=0}^{N-1}
     \left\| \left[ \gamma_\Lambda^j \big( 1_{\Lambda \setminus \Lambda_0} \otimes a_{\Lambda_0} \big),
                     b_\Lambda \right] \right\|_\Lambda.
\end{align}
For large $\Lambda$, most of the shifts $\gamma_\Lambda^j$ move the support of $a_{\Lambda_0}$ away 
from the support of $b_{\Lambda}$. More precisely, at most $|\Lambda_0| + |\Lambda'|$ terms in the 
sum correspond to non-disjoint supports, and hence to potentially non-zero commutators. For all other 
terms, the supports are disjoint, so the commutator vanishes. Thus,
\begin{equation}
  \| [a_\Lambda, b_\Lambda] \|_\Lambda
  \leq \frac{|\Lambda_0|+|\Lambda'|}{N}\; 2 \|a_{\Lambda_0}\|\, \|b_\Lambda\|_\Lambda.
\end{equation}
Since $(b_\Lambda)_\Lambda$ is a fixed local observable, $\|b_\Lambda\|_\Lambda$ is uniformly bounded. 
Hence,
$$
\| [a_\Lambda, b_\Lambda] \|_\Lambda=O(1/N),
$$
and thus
\begin{equation}
  \lim_{\Lambda \nearrow \Gamma} \| [a_\Lambda, b_\Lambda] \|_\Lambda = 0.
\end{equation}
This shows that every $\gamma$-sequence $(a_\Lambda)_\Lambda$ belongs to $\mathcal{C}^\infty$. Passing to equivalence classes yields
\begin{equation}
  [a_\Lambda]_\Lambda = q((a_\Lambda)_\Lambda) \in [\mathcal{C}]^\infty.
\end{equation} 
 \end{proof}

\section{$C^*$-algebraic construction of global observables in classical statistical mechanics}
Consider again the discrete set $\Gamma \subset \mathbb{R}^d$. To each $x \in \Gamma$ we associate a unital commutative $C^*$-algebra endowed with a Poisson bracket, $(\mathcal{A}_x, \{\cdot,\cdot\})$. This structure plays an important role in equilibrium thermodynamics within classical statistical mechanics, as it provides a natural framework for formulating the classical KMS condition \cite{DV1}.

For any finite subset $\Lambda \Subset \Gamma$, we denote the tensor product $\bigotimes_{x \in \Lambda} \mathcal{A}_x$ by $\mathcal{A}_\Lambda$, and equip it with the local supremum norm $\|\cdot\|_\Lambda$.  To define the Poisson bracket, we first fix a dense Poisson $*$-subalgebra $\dot{\mathcal{A}_x} \subset \mathcal{A}_x$ and then define $\dot{\mathcal{A}}_\Lambda$ accordingly for each finite subset $\Lambda \Subset \Gamma$. Then, in a fashion similar to above, the algebra
$$
\dot{\mathcal{D}}^\infty:=\left\{ (a_\Lambda)_\Lambda \in \prod_{\Lambda\Subset\Gamma} \dot{\mathcal{A}}_\Lambda : \ \forall (b_\Lambda)_\Lambda \in \dot{ \mathcal{A}}^\infty \ \lim_{\Lambda \nearrow \Gamma} \|\{a_\Lambda, b_\Lambda\}_\Lambda\|_\Lambda = 0 \right\},
$$
where $ \dot{\mathcal{A}}^\infty$ denotes the Poisson algebra of all local sequences, is a commutative $*$-subalgebra of $\prod_{\Lambda\Subset\Gamma} \mathcal{A}_\Lambda$.  
For a finite region $\Lambda \Subset \Gamma$, we define
\[
\dot{\mathcal{D}}_\Lambda^\infty :=
\Big\{ (a_{\Lambda'})_{\Lambda'} \in \prod_{\Lambda' \Subset \Gamma} \dot{\mathcal{A}}_{\Lambda'} \;:\; 
\forall b_\Lambda \in \dot{\mathcal{A}}_\Lambda \ \lim_{\Lambda' \nearrow \Gamma} \| \{ a_{\Lambda'}, b_\Lambda \}_{\Lambda'} \|_{\Lambda'} = 0
\Big\},
\]
where $b_\Lambda$ is embedded canonically into $\mathcal{A}_{\Lambda'}$ for $\Lambda' \supseteq \Lambda$. Analogous to the algebra $\mathcal{C}^\infty$, it follows that 
\[
\dot{\mathcal{D}}^\infty = \bigcap_{\Lambda \Subset \Gamma} \dot{\mathcal{D}}_\Lambda^\infty.
\]
Moreover, we can complete $\dot{\mathcal{D}}^\infty$ with respect to the norm $\|\cdot\|_\infty:=\sup_{\Lambda\Subset\Gamma}\|\cdot\|_\Lambda$, yielding the $C^*$-algebra
$$\mathcal{D}^\infty:=\overline{\dot{\mathcal{D}}^\infty}^{\|\cdot\|_\infty}.
$$ 
The pertinent quotient 
\[
[\mathcal{D}]^\infty:= \mathcal{D}^\infty/
\bigg{(}\bigoplus_{\Lambda\Subset\Gamma}\mathcal{A}_\Lambda\cap \mathcal{D}^\infty\bigg{)},
\]
with $\bigoplus_{\Lambda\Subset\Gamma}\mathcal{A}_\Lambda\cap \mathcal{D}^\infty$ the ideal of vanishing sequences,
is the classical analog of $[\mathcal{C}]^\infty$, and forms a  $C^*$-subalgebra of the quotient $C^*$-algebra $[\mathcal{A}]_\sim$.

\begin{remark}
\emph{We recall the tail-$\sigma$-algebra from classical statistical mechanics: 
\begin{equation} \label{eq:tail_sigma}
\mathscr{T}_\infty := \bigcap_{\Lambda \Subset \Gamma} \mathscr{F}_{\Lambda^c},
\end{equation}
where $\mathscr{F}_{\Lambda^c}$ is the
$\sigma$-algebra of events that only depend on ``spins'' located outside $\Lambda$. The algebra $\mathcal{D}^\infty$ resembles its algebraic analog, canonically formalized in a $C^*$-algebraic manner.
\hfill$\diamond$}
\end{remark}


    
    
    

\paragraph{Acknowledgments.}
The author is supported by a  Marie Sk\l odowska-Curie Actions (MSCA) Fellowship for postdoctoral research, under project number 101210672--MACROQC and acknowledges the support of the European Commission. 
Moreover, he gratefully acknowledges Roberto Beneduci for the invitation to IQSA 2025 and the enjoyable hospitality, and thanks Carmine De Rosa for useful discussions on this topic.

\paragraph{Data availability statement.} Data sharing is not applicable to this article as no new data were created or analysed in this study.

\paragraph{Funding Information.} Not applicable.


\begin{thebibliography}{999}

	

 \bibitem{Blackadar_2006}
	Blackadar B.,
	\textit{Operator algebras -Theory of $C^*$-algebras and von Neumann algebras},
	Springer-Verlag Berlin Heidelberg
	\href{https://doi.org/10.1007/3-540-28517-2}
	{(2006)}.

  \bibitem{DRW}   Duffield, N.G., Roos, H., , Werner, R.F, Macroscopic limiting dynamics of a class of inhomogeneous mean field quantum systems. {\em Annales de L'Institut Henri Poincar\'{e}} Physique Th\'{e}orique 56, 143-186 (1992).


 
    \bibitem{Bratteli_Robinson_1981_I}
	Bratteli O., Robinson D. W.,
	Operator Algebras and Quantum Statistical Mechanics. Vol. I: Equilibrium States, Models in Statistical Mechanics, Springer, (1981).
 
    \bibitem{Bratteli_Robinson_1981_II}
	Bratteli O., Robinson D. W.,
	Operator Algebras and Quantum Statistical Mechanics. Vol. II: Equilibrium States, Models in Statistical Mechanics, Springer, (1981).

    \bibitem{BruPedra} J.-B. Bru, W. de Siqueira Pedra, Non-cooperative Equilibria of Fermi Systems With Long Range Interactions, Memoirs of the AMS, Vol. 224, No. 1052 (2013).

\bibitem{DVP} N. Drago, L. Pettinari, C.J.F. van de Ven, Classical and quantum KMS states on spin lattice systems, {\em Communications in Mathematical Physics} Vol. 406, 163 (2025).

 \bibitem{DV2} N. Drago, C.J.F. van de Ven, Strict deformation quantization and local spin interactions, {\em Communication in Mathematical Physics}, Vol 405, 14 (2024).

 
\bibitem{DV1} N. Drago, C.J.F. van de Ven, DLR-KMS correspondence on lattice spin systems, {\em Letters in Mathematical Physics}, Vol 113, 88 (2023).

    \bibitem{Enter} A. C. D. van Enter, .R.  Fernandez, and A. D. Sokal, Regularity Properties and Pathologies of Position- 
Space Renormalization-Group Transformations: 
Scope and Limitations of Gibbsian Theory, {\em Journal of Statistical Physics}, Vol. 72, Nos. 5/6, (1993).

\bibitem{Georgii}  H. O. Georgii, Gibbs measures and phase transitions, De Gruyter Studies in Mathematics Vol. 9. Berlin: de Gruyter 1988. second edition, (2011).

\bibitem{Hepp} K. Hepp, Quantum theory of measurement and
 macroscopic observables, {\em Helvetica Physica Acta} 48 (1972).

 \bibitem{Israel}  R. B. Israel, Convexity in the Theory of Lattice Gases, Princeton University Press,  Princeton, New Jersey, (1979).

\bibitem{LR} O.E. Landford , D. Ruelle, Observables at Infinity and States  Short Range Correlations in Statistical Mechanics 
{\em Commun. math. Phys.} 13, 194--215 (1969).
    
    \bibitem{Landsman_2017}
    Landsman N. P.,
    Foundations of quantum theory: From classical concepts to operator algebras,
    Springer Nature (2017).

\bibitem{LMV} Landsman K., Moretti V., Van de Ven, C.J.F., Strict deformation quantization of the state space of $M_k(\mathbb{C})$ with applications to the Curie-Weiss model, {\em Reviews in Mathematical Physics} Vol. 32, No. 10, 2050031 (2020).

\bibitem{Raggio} G.A. Raggio and R.F. Werner, Quantum statistical mechanics of general mean
field systems, {\em Helvetica Physica} Acta 62, 980--1003 (1989).
   
    

\end{thebibliography}
\end{document}